\documentclass{llncs_dagger}
\usepackage[english]{babel}
\usepackage[utf8]{inputenc}
\usepackage{amssymb,upgreek} 
\usepackage{amsmath,xspace}
\usepackage{color}
\usepackage{graphicx}
\usepackage{hyperref}

\newcommand{\cG}{{\mathcal G}}

\newcommand{\NP}{{\sf NP}}

\newcommand{\FPT}{{\sf FPT}}

\newcommand{\W}{{\sf W}}
\newcommand{\tw}{{\textup{tw}}}

\DeclareMathOperator{\drm}{drm}

\newcommand{\fullc}{\ensuremath\xrightarrow{_B}\xspace}
\newcommand{\surjc}{\ensuremath\xrightarrow{_S}\xspace}
\newcommand{\parc}{\ensuremath\xrightarrow{_I}\xspace}

\iffalse
\newcommand{\Hom}{{\sc Homomorphism}}

\newcommand{\xLBHom}[1]{{\sc Locally Bijective $#1$-Homomorphism}}
\newcommand{\xLIHom}[1]{{\sc Locally Injective $#1$-Homomorphism}}
\newcommand{\xLSHom}[1]{{\sc Locally Surjective $#1$-Homomorphism}}
\else
\newcommand{\Hom}{{\sc Hom}}
\newcommand{\LBHom}{{\sc LBHom}}
\newcommand{\LIHom}{{\sc LIHom}}
\newcommand{\LSHom}{{\sc LSHom}}
\newcommand{\xLBHom}[1]{{\sc $#1$-LBHom}}
\newcommand{\xLIHom}[1]{{\sc $#1$-LIHom}}
\newcommand{\xLSHom}[1]{{\sc $#1$-LSHom}}
\fi

{\bf}{\it}
 
\begin{document}
\title{Locally Constrained Homomorphisms on Graphs of Bounded Treewidth and Bounded Degree\thanks{This paper is supported by the Natural Sciences Engineering Research Council of Canada (NSERC), the Research Council of Norway (197548/F20), EPSRC (EP/G043434/1) and the Royal Society (JP100692). An extended abstract of it appeared in the proceedings of FCT 2013, LNCS 8070: 121-132.}}

\titlerunning{On constrained homomorphisms}

\author{ 
Steven Chaplick \inst{1,}\thanks{Supported by the ESF GraDR EUROGIGA grant as project GACR GIG/11/E023 and the NSERC grants of: K. Cameron and C. Ho\`ang  (Wilfrid Laurier University), D. Corneil (University of Toronto), and P. Hell (Simon Fraser University).}
\and
Ji\v{r}\'i Fiala \inst{2,}\thanks{Supported by M\v{S}MT \v{C}R grant LH12095 and GA\v{C}R grant P202/12/G061.}
\and
Pim van~'t Hof \inst{3}
\and\\
Dani\"el Paulusma \inst{4}
\and
Marek Tesa\v{r} \inst{2}
}

\institute{
Institut fur Mathematik, TU Berlin, Germany\\
\texttt{chaplick@math.tu-berlin.de} 
Department of Applied Mathematics, Charles University, Prague, Czech Republic\\
\texttt{\{fiala,tesar\}@kam.mff.cuni.cz}
\and
Department of Informatics, University of Bergen, Norway\\
\texttt{pim.vanthof@ii.uib.no}
\and
School of Engineering and Computing Sciences, Durham University, UK\\
\texttt{daniel.paulusma@durham.ac.uk}}

\maketitle

\begin{abstract}
A homomorphism from a graph $G$ to a graph $H$ is locally bijective, surjective, or injective if its restriction to the neighborhood of every vertex of $G$ is bijective, surjective, or injective, respectively. 
We prove that the problems of testing whether a given graph $G$ allows a homomorphism to a given graph $H$ that is locally bijective, surjective, or injective, respectively, are \NP-complete, even 
when $G$ has pathwidth at most $5$, $4$, or $2$, respectively, 
or when both $G$ and $H$ have  maximum degree~$3$.
We complement these hardness results by showing that the three problems are polynomial-time solvable if $G$ has bounded treewidth and in addition $G$ or $H$ has bounded maximum degree.
\end{abstract}


\section{Introduction}

All graphs considered in this paper are finite, undirected, and have neither self-loops nor multiple edges.
A {\it graph homomorphism} from a graph $G=(V_G,E_G)$ to a graph $H=(V_H,E_H)$ is a mapping $\varphi: V_G \to V_H$ that maps adjacent vertices of $G$ to adjacent
vertices of $H$, i.e., $\varphi(u)\varphi(v)\in E_H$ whenever $uv\in E_G$. 
The notion of a graph homomorphism is well studied in the literature due to its many practical and theoretical applications;
we refer to the textbook of Hell and Ne\v{s}et\v{r}il~\cite{HN04} for a survey.

We write $G\to H$ to indicate the existence of a homomorphism from $G$ to $H$.
We call $G$ the {\it guest graph} and $H$ the {\it host graph}. We denote the vertices of $H$ by $1,\ldots,|H|$ and call them {\it colors}. 
The reason for doing this is that graph homomorphisms generalize graph colorings: there exists a homomorphism from a graph $G$ to a complete graph on $k$ vertices if and only if $G$ is $k$-colorable.
The problem of testing whether $G\to H$ for two given graphs $G$ and $H$ is called the \Hom{} problem. If only the guest graph is part of the input and the host graph is {\it fixed}, i.e., not part of the input, then this problem is denoted as $H$-\Hom{}.
The classical result in this area is the Hell-Ne\v{s}et\v{r}il dichotomy
theorem which states that  $H$-\Hom{} is solvable in polynomial time if $H$ is bipartite, and \NP-complete otherwise~\cite{HN90}.

We consider so-called {\it locally constrained} homomorphisms.
The {\it neighborhood} of a vertex $u$ in a graph $G$ is denoted
$N_G(u)=\{v\in V_G\; |\; uv\in E_G\}$. 
If for every $u\in V_G$ the restriction of $\varphi$ to the neighborhood
of $u$, i.e., the mapping $\varphi_u:N_G(u)\to N_H(\varphi(u))$, is injective, bijective, or surjective, then $\varphi$ is said to be {\it locally injective}, {\it locally bijective},  or {\it locally surjective}, respectively. 
Locally bijective homomorphisms are also called {\it graph coverings}. They
originate from topological graph theory~\cite{Bi74,Ma67} and have applications in distributed 
computing~\cite{An80,AG81,Bo89}
and in constructing highly transitive regular graphs~\cite{Bi82}. 
Locally injective homomorphisms are also called {\it partial graph coverings}. They
have applications in models of telecommunication~\cite{FK02} and in
distance constrained labeling~\cite{FKK01}. Moreover, they are used 
as indicators of the existence of homomorphisms  of 
derivative graphs~\cite{Ne71}.
Locally surjective homomorphisms are also called {\it color dominations}~\cite{KT00}. In addition they are known as  {\it role
assignments} due to their applications in social science~\cite{EB91,PR01,RS01}.
Just like locally bijective homomorphisms they also have applications in distributed computing~\cite{CMZ04}.

If there exists a homomorphism from a graph $G$ to a graph $H$ that is locally bijective, locally injective, or locally surjective,  
respectively, then we write $G\fullc H$, $G\parc H$, and $G\surjc H$,  
respectively.
We denote the decision problems that are to test whether $G\fullc H$, $G\parc H$, or $G\surjc H$ 
for two given graphs $G$ and $H$ by \LBHom{}, \LIHom{} and \LSHom{}, respectively.
All three problems are known to be \NP-complete when both guest and host graphs are given as input 
(see below for details), and attempts have been made to classify their computational complexity when only the guest graph belongs to the input and the host graph is fixed.
The corresponding problems are denoted by \xLBHom{H}, \xLIHom{H}, and \xLSHom{H}, respectively.
The \xLSHom{H} problem is polynomial-time solvable 
either if $H$ has no edge or if $H$ is bipartite and has at least one connected component isomorphic to an edge; 
in all other cases \xLSHom{H} is \NP-complete, 
even when the guest graph belongs to the class of bipartite graphs~\cite{FP05}.
The complexity classification of \xLBHom{H} and \xLIHom{H} is still open, although many partial results are known for both problems; we refer to the papers~\cite{AFS91,BLT11,FK02,FKP08,KPT97,KPT98,LT10} 
and to the survey by Fiala and Kratochv\'{i}l~\cite{FK08} for both \NP-complete and
polynomially solvable cases.

Instead of fixing the host graph, another natural restriction is to only take guest graphs from a special graph class. 
Heggernes et al.~\cite{HHP12} proved that \LBHom{} is {\sc Graph Isomorphism}-complete when the guest graph is chordal, and polynomial-time solvable when the guest graph is interval.
In contrast, \LSHom{} is \NP-complete when the guest graph is chordal 
and polynomial-time solvable when the guest graph is proper interval, whereas 
\LIHom{} is \NP-complete even for guest graphs that are proper interval~\cite{HHP12}.
It is also known that the problems \LBHom{} and \LSHom{} are polynomial-time solvable when the guest graph is a tree~\cite{FP10}.

In this paper we focus on the following line of research.
The {\it core} of a graph $G$ is a subgraph $F$ of $G$ such that $G\to F$ and there is no proper subgraph $F'$ of $F$ with $G\to F'$.
It is known that the core of a graph is unique up to isomorphism~\cite{HN92}. 
Dalmau, Kolaitis and Vardi~\cite{DKV02} proved that the \Hom{} problem is polynomial-time solvable when the guest graph belongs to any fixed 
class of graphs whose cores have bounded treewidth. In particular, this result implies an earlier result that \Hom{} is polynomial-time solvable when the guest graph has bounded treewidth~\cite{CR97,Fr90}. Grohe~\cite{Gr07} strengthened the result of Dalmau et al.~\cite{DKV02} by proving that under a certain complexity assumption, 
namely $\FPT\neq\W[1]$, the \Hom{} problem can be solved in polynomial time if and only if this condition holds.

It is a natural question whether the above results of Dalmau et al.~\cite{DKV02} and Grohe~\cite{Gr07} 
remain true when we consider locally constrained homomorphisms instead of general homomorphisms.
We can already conclude from known results that this is not the case for 
locally surjective homomorphisms.
Recall that \xLSHom{H} is \NP-complete even for bipartite guest graphs if $H$ contains at least one edge and is either non-bipartite or does not contain a connected component isomorphic to an edge~\cite{FP05}.
The core of every bipartite graph with at least one edge is 
an edge, and consequently, has treewidth~1. This means that bipartite graphs form a class of graphs whose cores have bounded treewidth.
Due to this negative answer, we pose the following (weaker) question instead:

\medskip
\noindent
{\it Are \LBHom, \LIHom{} and \LSHom{} polynomial-time solvable when the guest graph belongs to a class of bounded treewidth?}

\medskip
\noindent
This question is further motivated by two known results, namely that \LBHom{} and \LSHom{} can both be solved in polynomial time if the guest graph is a tree, that is, has treewidth 1~\cite{FP10}.

\subsection*{Our Contribution}
In Section~\ref{s-npcom}, we provide a negative answer to this question by showing that 
the problems \LBHom{}, \LSHom{} and \LIHom{} are \NP-complete already in the restricted case where the guest graph has pathwidth 
at most $5$, $4$ or $2$, respectively. 
We also show that the three problems are \NP-complete 
even if both the guest graph and the host graph have maximum degree~$3$.
The latter result shows that locally constrained homomorphisms problems behave more like unconstrained homomorphisms on graphs of bounded degree than on graphs of bounded treewidth, as it is known that, for example, $C_5$-{\sc Hom} is \NP-complete on subcubic graphs~\cite{GHN00}.

On the positive side, in Section~\ref{s-pol}, we show that all three problems can be solved in polynomial time if we bound the treewidth of the guest graph and at the same time bound the maximum degree of the guest graph or the host graph. 
Because a graph class of bounded maximum degree has bounded treewidth if and only if it has bounded clique-width~\cite{GW00}, 
all three problems are also polynomial-time solvable when  we bound the clique-width and the maximum degree of the guest graph.
In Section~\ref{s-pol} we also show that \LIHom{} can be solved in polynomial time when the guest graph has treewidth~1, which is best possible given the hardness result for \LIHom{} shown in Section~\ref{s-npcom}.

In Section~\ref{s-con} we state some relevant open problems.

\section{Preliminaries}

Let $G$ be a graph. 
The {\em degree} of a vertex $v$ in $G$ is denoted by $d_G(v)=|N_G(v)|$, and $\Delta(G)=\max_{v\in V_G} d_G(v)$ denotes the maximum degree of $G$. Let $\varphi$ be a homomorphism from $G$ to a graph $H$. Moreover, let $G'$ be an induced subgraph of $G$, and let $\varphi'$ be a homomorphism from $G'$ to $H$. We say that $\varphi$ {\em extends} (or, equivalently, is an {\em extension} of) $\varphi'$ if $\varphi(v)=\varphi'(v)$ for every $v\in V_{G'}$.

A \emph{tree decomposition} of $G$ is a tree $T=(V_T,E_T)$, 
where the elements of $V_T$, called the \emph{nodes} of $T$, are subsets of $V_G$ such that the following three conditions are satisfied:
\begin{itemize}
\item [1.] for each vertex $v\in V_G$, there is a node $X\in V_T$ with $v\in X$,
\item [2.] for each edge $uv\in E_G$, there is a node $X\in V_T$ with $\{u,v\}\subseteq X$,
\item [3.] for each vertex $v\in V_G$, the set of nodes $\{X \mid v\in X\}$ 
induces a connected subtree of $T$. 
\end{itemize}
The \emph{width} of a tree decomposition $T$ is the size of a largest node $X$ minus~one. The {\it treewidth} of $G$, denoted by $\tw(G)$, is the minimum width over all possible tree decompositions of $G$.
A {\it path decomposition} of $G$ is  a tree decomposition $T$ of $G$ where $T$ is a path.  
The  {\it pathwidth} of $G$ 
is the minimum width over all possible path decompositions of $G$.
By definition, the pathwidth of $G$ is at least as high as its treewidth.
A tree decomposition $T$ is {\it nice} \cite{Kloks94} if $T$ is a binary tree, rooted in a root $R$
such that the nodes of $T$ belong to one of the following four types: 
\begin{itemize}
\item [1.] a \emph{leaf node} $X$ is a leaf of $T$,
\item [2.] an \emph{introduce node} $X$ has one child $Y$ and $X=Y\cup \{v\}$ for some vertex $v\in V_G\setminus Y$,
\item [3.] a \emph{forget node} $X$ has one child $Y$ and $X=Y\setminus \{v\}$ for some vertex $v\in Y$,
\item [4.] a \emph{join node} $X$ has two children $Y,Z$ satisfying $X=Y=Z$.
\end{itemize}  

\section{NP-Completeness Results}
\label{s-npcom}

For the NP-hardness results in Theorem~\ref{t-np} below
we use a reduction from the 3-{\sc Partition} problem. 
This problem takes as input a multiset $A$ of $3m$ integers, denoted in the sequel by $\{a_1,a_2,\ldots,a_{3m}\}$,
and a positive integer $b$, such that $\frac{b}4<a_i<\frac{b}2$ for all $i\in \{1,\ldots,3m\}$ and $\sum_{1\leq i\leq 3m} a_i=mb$. 
The task is to determine whether $A$ can be partitioned into $m$ disjoint sets $A_1,\ldots,A_m$ such that $\sum_{a\in A_i} a=b$ 
for all $i\in \{1,\ldots,m\}$. Note that the restrictions on the size of each element in $A$ implies that each set $A_i$ 
in the desired partition must contain exactly three elements, which is why such a partition $A_1,\ldots,A_m$ is 
called a {\em $3$-partition} of $A$. The 3-{\sc Partition} problem is strongly \NP-complete~\cite{GareyJ79}, 
i.e., it remains \NP-complete even if the problem is encoded in unary. 

\setcounter{footnote}{0}
\begin{theorem}\label{t-np}
The following three statements hold:
\begin{itemize}
\item[(i)] \LBHom{} is \NP-complete on input pairs $(G,H)$ where $G$ has pathwidth at most~$5$ and $H$ has pathwidth at most~$3$;
\item[(ii)] \LSHom{} is \NP-complete on input pairs $(G,H)$ where $G$ has pathwidth at most~$4$ and $H$ has pathwidth at most~$3$;
\item[(iii)] \LIHom{} is \NP-complete on input pairs $(G,H)$ where $G$ has pathwidth at most $2$ and $H$ has pathwidth at most~$2$.
\end{itemize}
\end{theorem}

\begin{proof}
First note that all three problems are in \NP. We prove each statement separately starting with statement (i).

\begin{figure}[htb]
\centering
\includegraphics[scale=.85]{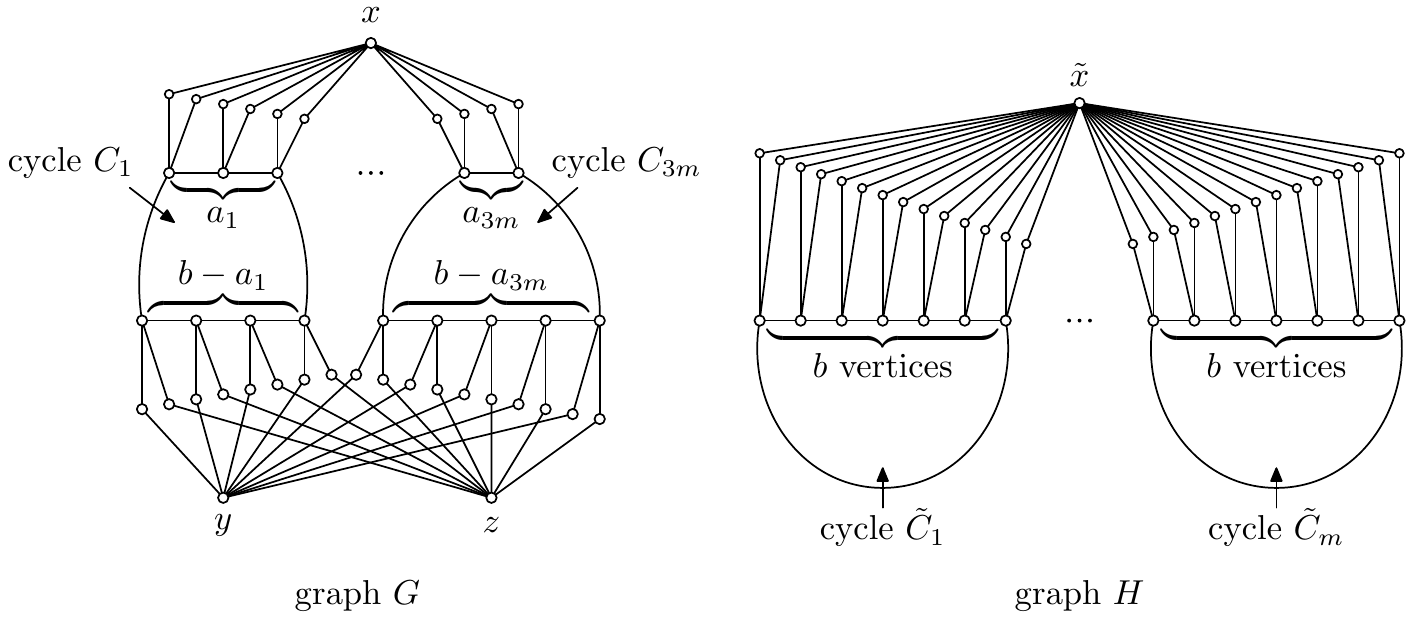}
\caption{A schematic illustration of the graphs $G$ and $H$ that are constructed from a given instance $(A,b)$ of $3$-{\sc Partition} in the proof of statement (i) in Theorem~\ref{t-np}.
See also Figure~\ref{fig:detail} for a more detailed illustration of the ``leftmost'' part of $G$ and the ``rightmost'' part of $H$, including more labels.}
\label{fig:reduction_bijective}
\end{figure}

Note that \LBHom{} is in \NP. Given an instance $(A,b)$ of 3-{\sc Partition}, we construct two graphs $G$ and $H$ as follows; 
see Figures~\ref{fig:reduction_bijective} and~\ref{fig:detail} for some helpful illustrations. 
The construction of $G$ starts by taking $3m$ disjoint cycles $C_1,\ldots,C_{3m}$ of length $b$, one for each element of $A$. 
For each $i\in \{1,\ldots,3m\}$, the vertices of $C_i$ are labeled $u^i_{1},\ldots,u^i_b$ and we add, for each $j\in \{1,\ldots,b\}$, 
two new vertices $p^i_j$ and $q^i_j$ as well as two new edges $u^i_jp^i_j$ and $u^i_jq^i_j$. 
We then add three new vertices $x$, $y$ and $z$. Vertex $x$ is made adjacent to vertices $p^i_1,p^i_2\ldots,p^i_{a_i}$ 
and $q^i_1,q^i_2\ldots,q^i_{a_i}$ for every $i\in \{1,\ldots,3m\}$. 
Finally, the vertex $y$ is made adjacent to every vertex $p^i_j$ that is not adjacent to $x$, and the vertex $z$ is made adjacent 
to every vertex $q^i_j$ that is not adjacent to $x$. This finishes the construction of $G$. 

To construct $H$, we take $m$ disjoint cycles $\tilde{C}_1,\ldots,\tilde{C}_{m}$ of length $b$, 
where the vertices of each cycle $\tilde{C}_i$ are labeled $\tilde{u}^i_1,\ldots,\tilde{u}^i_b$. 
For each $i\in \{1,\ldots,m\}$ and $j\in \{1,\ldots,b\}$, we add two vertices $\tilde{p}^i_j$ 
and $\tilde{q}^i_j$ and make both of them adjacent to $\tilde{u}^i_j$. 
Finally, we add a vertex $\tilde{x}$ and make it adjacent to each of the vertices $\tilde{p}^i_j$ and $\tilde{q}^i_j$. 
This finishes the construction of $H$.

\begin{figure}[htb]
\centering
\includegraphics[scale=.85]{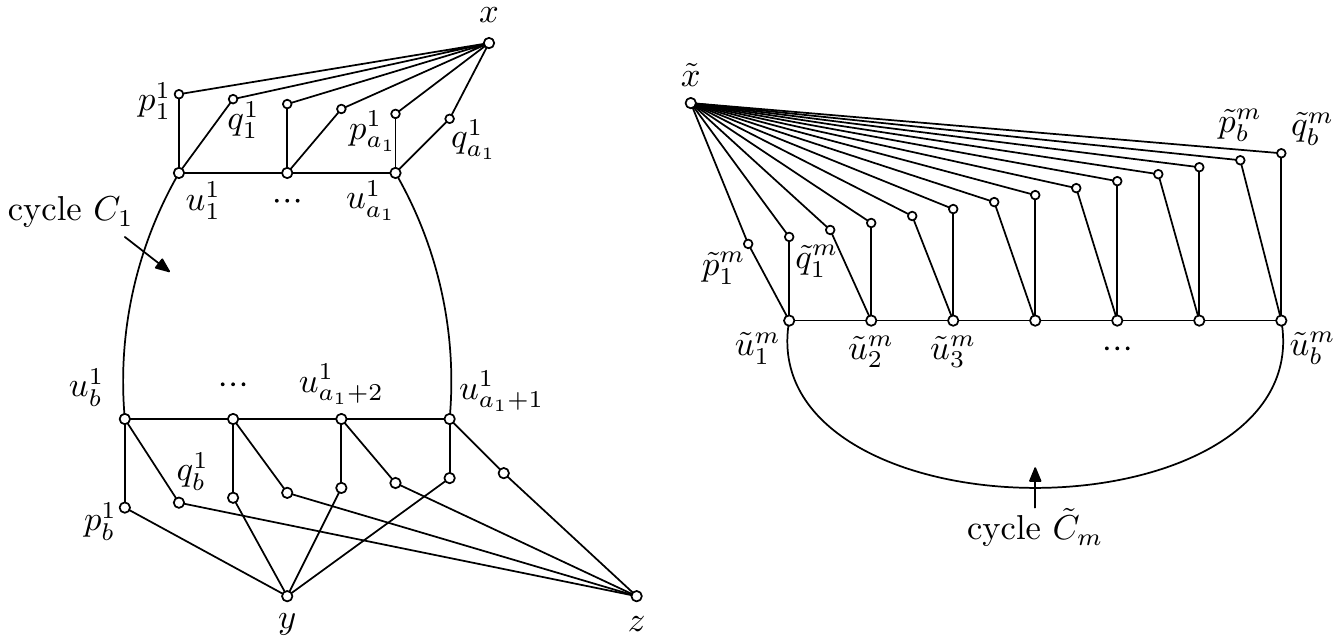}
\caption{More detailed illustration of parts of the graphs $G$ and $H$ in Figure~\ref{fig:reduction_bijective}.}
\label{fig:detail}
\end{figure}

We now show that there exists a locally bijective homomorphism from $G$ to $H$ if and only if $(A,b)$ is a yes-instance of 3-{\sc Partition}. 

Let us first assume that there exists a locally bijective homomorphism $\varphi$ from $G$ to $H$. 
Since $\varphi$ is a degree-preserving mapping, we must have $\varphi(x)=\tilde{x}$. 
Moreover, since $\varphi$ is locally bijective, the restriction of $\varphi$ to $N_G(x)$ is a bijection from $N_G(x)$ to $N_H(\tilde{x})$. 
Again using the definition of a locally bijective mapping, this time considering the neighborhoods of the vertices in $N_H(\tilde{x})$, 
we deduce that there is a bijection from the set 
$N^2_G(x):= \{u^i_j \mid 1\leq i\leq 3m, 1\le j \le a_i \}$, i.e., from the set of vertices in $G$ at distance $2$ from $x$, to the set 
$N^2_H(\tilde{x}):=\{\tilde{u}^k_j \mid 1\leq k\leq m, 1\le j \le b \}$ 
of vertices that are at distance $2$ from $\tilde{x}$ in $H$. 
For every $k\in \{1,\ldots, m\}$, we define a set $A_k\subseteq A$ such that $A_k$ contains element $a_i\in A$ if and only if 
$\varphi(u^i_1)\in \{\tilde{u}^k_1,\ldots,\tilde{u}^k_b\}$. 
Since $\varphi$ is a bijection from $N^2_G(x)$ to $N^2_H(\tilde{x})$, the sets $A_1,\ldots,A_m$ are disjoint; 
moreover each element $a_i\in A$ is contained in exactly one of them. 
Observe that the subgraph of $G$ induced by $N^2_G(x)$ is a disjoint union of $3m$ paths of lengths 
$a_1,a_2,\ldots,a_{3m}$, respectively, while the subgraph of $H$ induced by $N^2_H(\tilde{x})$ 
is a disjoint union of $m$ cycles of length $b$ each.                                                            
The fact that $\varphi$ is a homomorphism and therefore never maps adjacent vertices of $G$ to non-adjacent vertices in $H$ 
implies that $\sum_{a\in A_i} a=b$ for all $i\in \{1,\ldots,m\}$. Hence $A_1,\ldots,A_m$ is a $3$-partition of $A$.

For the reverse direction, suppose there exists a 3-partition $A_1,\ldots,A_m$ of $A$. We define a mapping $\varphi$ as follows. 
We first set $\varphi(x)=\varphi(y)=\varphi(z)=\tilde{x}$. Let $A_i=\{a_r,a_s,a_t\}$ be any set of the 3-partition. 
We map the vertices of the cycles $C_r,C_s,C_t$ that are at distance $2$ from $x$ to the vertices of the cycle $\tilde{C}_i$ in the following way:
$\varphi(u^r_j)=\tilde{u}^i_j$ for each $j\in \{1,\ldots,a_r\}$, $\varphi(u^s_j)=\tilde{u}^i_{a_r+j}$ for each $j\in \{1,\ldots,a_s\}$, and $\varphi(u^t_j)=\tilde{u}^i_{a_r+a_s+j}$ for each $j\in \{1,\ldots,a_t\}$. 
The vertices of $C_r$, $C_s$ and $C_t$ that are at distance more than $2$ from $x$ in $G$ are mapped to vertices of $\tilde{C}_i$ such that the vertices of $C_r$, $C_s$ and $C_t$ appear in the same order as their images on $\tilde{C}_i$. 
In particular, we set $\varphi(u^r_j)=\tilde{u}^i_j$ for each $j\in \{a_r+1,\ldots,b\}$; the vertices of the cycles $C_s$ and $C_t$ that are at distance more than~$2$ from $x$ are mapped to vertices of $\tilde{C}_i$ analogously. 
After the vertices of the cycles $C_1,\ldots,C_{3m}$ have been mapped in the way described above, it remains to map the vertices $p^i_j$ and $q^i_j$ for each $i\in \{1,\ldots,3m\}$ and $j\in \{1,\ldots,b\}$. 

Let $p^i_j,q^i_j$ be a pair of vertices in $G$ that are adjacent to $x$, and let $u^i_j$ be the second common neighbor of $p^i_j$ and $q^i_j$. Suppose $\tilde{u}^k_\ell$ is the image of $u^i_j$, i.e., suppose that $\varphi(u^i_j)=\tilde{u}^k_\ell$. Then we map $p^i_j$ and $q^i_j$ to $\tilde{p}^k_\ell$ and $\tilde{q}^k_\ell$, respectively. We now consider the neighbors of $y$ and $z$ in $G$. By construction, the neighborhood of $y$ consists of the $2mb$ vertices in the set 
$\{p^i_j \mid a_{i+1}\leq j \leq b\}$, while $N_G(z)=\{q^i_j \mid a_{i+1}\leq j \leq b\}$. 

Observe that $\tilde{x}$, the image of $y$ and $z$, is adjacent to two sets of $mb$ vertices: one of the form $\tilde{p}^k_\ell$, the other of the form $\tilde{q}^k_\ell$. 
Hence, we need to map half the neighbors of $y$ to vertices of the form $\tilde{p}^k_\ell$ and half the neighbors of $y$ to vertices of the form $\tilde{q}^k_\ell$ in order 
to make $\varphi$ a locally bijective homomorphism. The same should be done with the neighbors of $z$. For every vertex $\tilde{u}^k_\ell$ in $H$, we do as follows. 
By construction, exactly three vertices of $G$ are mapped to $\tilde{u}^k_\ell$, and exactly two of those vertices, say $u^i_j$ and $u^{g}_{h}$, are at distance $2$ from $y$ in $G$. 
We set $\varphi(p^i_j)=\tilde{p}^k_\ell$ and $\varphi(p^{g}_{h})=\tilde{q}^k_\ell$. 
We also set $\varphi(q^i_j)=\tilde{q}^k_\ell$ and $\varphi(q^{g}_{h})=\tilde{p}^k_\ell$.
This completes the definition of the mapping $\varphi$.

Since the mapping $\varphi$ preserves adjacencies, it clearly is a homomorphism. In order to show that $\varphi$ is locally bijective, we first observe that the degree of every vertex in $G$ is equal to the degree of its image in $H$; in particular, $d_G(x)=d_G(y)=d_G(z)=d_H(\tilde{x})=mb$. From the above description of $\varphi$ we get a bijection between the vertices of $N_H(\tilde{x})$ and the vertices of $N_G(v)$ for each $v\in \{x,y,z\}$. For every vertex $p^i_j$ that is adjacent to $x$ and $u^i_j$ in $G$, its image $\tilde{p}^k_\ell$ is adjacent to the images $\tilde{x}$ of $x$ and $\tilde{u}^k_\ell$ of $u^i_j$. For every vertex $p^i_j$ that is adjacent to $y$ (respectively $z$) and $u^i_j$ in $G$, its image $\tilde{p}^k_\ell$ or $\tilde{q}^k_\ell$ is adjacent to $\tilde{x}$ of $y$ (respectively $z$) and $\tilde{u}^k_\ell$ of $u^i_j$. Hence the restriction of $\varphi$ to $N_G(p^i_j)$ is bijective for every $i\in \{1,\ldots,3m\}$ and $j\in \{1,\ldots,b\}$, and the same clearly holds for the restriction of $\varphi$ to $N_G(q^i_j)$. The vertices of each cycle $C_i$ are mapped to the vertices of some cycle $\tilde{C}_k$ in such a way that the vertices and their images appear in the same order on the cycles. This, together with the fact that the image $\tilde{u}^k_\ell$ of every vertex $u^i_j$ is adjacent to the images $\tilde{p}^k_\ell$ and $\tilde{q}^k_\ell$ of the neighbors $p^i_j$ and $q^i_j$ of $u^i_j$, shows that the restriction of $\varphi$ to $N_G(u^i_j)$ is bijective for every $i\in \{1,\ldots,3m\}$ and $j\in \{1,\ldots,b\}$. We conclude that $\varphi$ is a locally bijective homomorphism from $G$ to $H$.

In order to show that the pathwidth of $G$ is at most~$5$, let us first consider the subgraph of $G$ depicted on the left-hand side of Figure~\ref{fig:detail}; we denote this subgraph by $L_1$, and we say that the cycle $C_1$ {\em defines} the subgraph $L_1$. The graph $L_1'$ that is obtained from $L_1$ by deleting vertices $x,y,z$ and edge $u^1_1 u^1_b$ is a caterpillar, i.e., a tree in which there is a path containing all vertices of degree more than $1$. Since caterpillars are well-known to have pathwidth $1$, graph $L_1'$ has a path decomposition 
$P_1'$ of width~$1$.
Starting with 
$P_1'$, we can now obtain a path decomposition of the graph $L_1$ by simply adding vertices $x$, $y$, $z$ and $u_1^1$ to each node of 
$P_1'$; this path decomposition has 
width~$5$. Every cycle $C_i$ in $G$ defines a subgraph $L_i$ of $G$ in the same way $C_1$ defines the subgraph $L_1$. Suppose we have constructed a path decomposition 
$P_i$ of width~$5$ of the subgraph $L_i$ for each $i\in \{1,\ldots,3m\}$ in the way described above. Since any two subgraphs $L_i$ and $L_j$ with $i\neq j$ have only the vertices $x,y,z$ in common, and these three vertices appear in all nodes of each of the path decompositions 
$P_i$, we can arrange the $3m$ path decompositions 
$P_1,\ldots,P_{3m}$ in such a way that we obtain a path decomposition 
$P$ of $G$ of width~$5$. Hence $G$ has pathwidth at most~$5$.
Similar but easier arguments can be used to show that $H$ has pathwidth
at most~$3$.

The \NP-hardness reduction for the locally bijective case can also be used to prove that \LIHom{} and \LSHom{} are \NP-hard for input pairs $(G,H)$ where $G$ has pathwidth at most~$5$ and $H$ has pathwidth at most~3.
This follows from the claim that $G\fullc H$ if and only if $G\surjc H$ if and only if $G\parc H$ for the gadget graphs $G$ and $H$ displayed in Figure~\ref{fig:reduction_bijective}. 
This claim can be seen as follows. First suppose that $G\fullc H$. Then, by definition, $G\surjc H$ and $G\parc H$.
Now suppose that $G\parc H$ or $G\surjc H$. Since it can easily be verified that
$$
\drm(G)=\drm(H)=
\begin{pmatrix}
0\; & 0\; & 2mb \\
0 & 2 & 2 \\
1 & 1 & 0 \\
\end{pmatrix},
$$
we can use Lemma~\ref{l-drm} (i) or (ii), respectively, to deduce that $G\fullc H$.
However, we can strengthen the hardness results for the locally surjective and injective cases by reducing the pathwidth of the guest graph 
to be at most
$4$ and~$2$, respectively, 
and in the latter case we can simultaneously reduce the pathwidth of the host graph to be at most~$2$, 
as claimed in statements (ii) and (iii) of Theorem~\ref{t-np}. In order to do so, we give the following alternative 
constructions below.

\begin{figure}[htb]
\centering
\includegraphics[scale=.85]{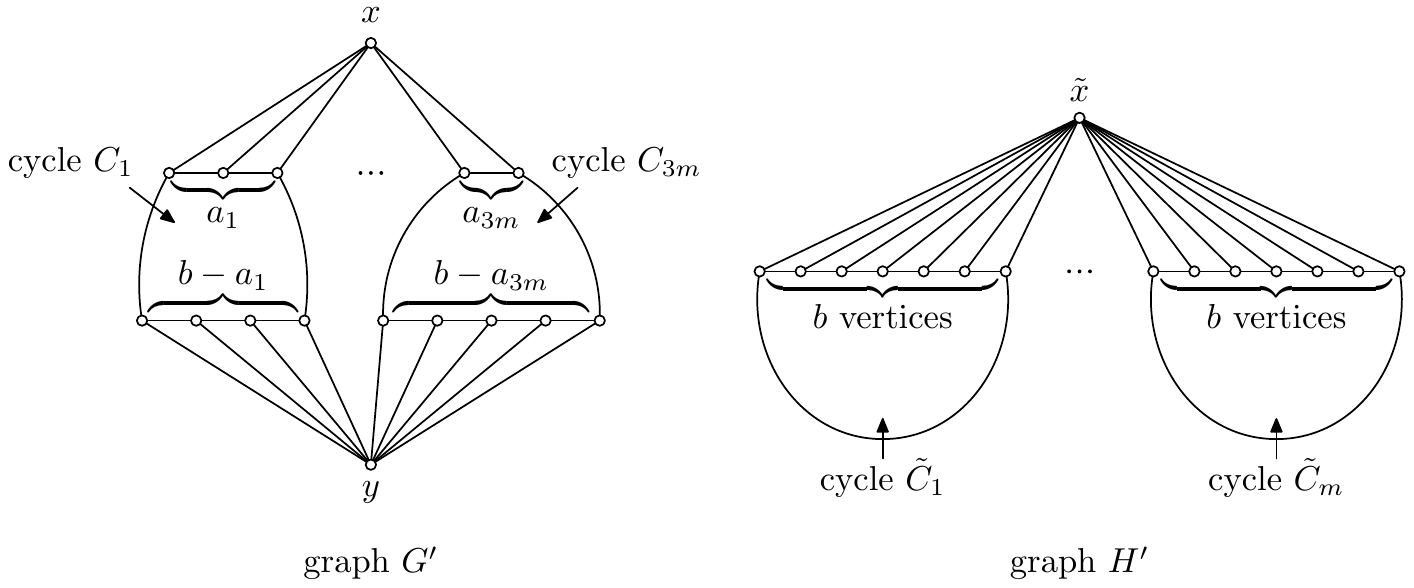}
\caption{A schematic illustration of the graphs $G'$ and $H'$ that are constructed from a given instance $(A,b)$ of $3$-{\sc Partition} in the proof of statement (ii) in Theorem~\ref{t-np}.
}
\label{fig:reduction_injective}
\end{figure}

The alternative hardness construction for \LSHom{} is similar to but easier than the construction for \LBHom{}; see Figure~\ref{fig:reduction_injective}. 
Let $(A,b)$ be an instance of $3$-{\sc Partition}. We construct a graph $G'$ by taking $3m$ disjoint cycles $C_1,\ldots,C_{3m}$ of length~$b$, and labeling the vertices of each cycle $C_i$ 
with labels $u^i_1,\ldots,u^i_B$ in the same way as we labeled the vertices of the cycles $C_i$ in the construction for \LBHom{} (see also Figure~\ref{fig:detail}). We then add two vertices $x$ and $y$. 
For every $i\in \{1,\ldots,3m\}$, we make $x$ adjacent to each of the vertices $u^i_1,u^i_2\ldots,u^i_{a_i}$, and $y$ is made adjacent to each of the vertices $u^i_{a_i+1},\ldots,u^i_B$. 
Graph $H'$ is obtained from the disjoint union of $m$ cycles $\tilde{C_1},\ldots,\tilde{C_m}$ of length $b$ by adding one universal vertex $\tilde{x}$. Using similar arguments as the ones used in the \NP-hardness proof of \LBHom{}, 
it can be shown that there exists a locally surjective homomorphism $\varphi$ from $G'$ to $H'$ if and only if $(A,b)$ is a yes-instance of $3$-{\sc Partition}. Such a homomorphism $\varphi'$ maps $x$ and $y$ to $\tilde{x}$, 
and maps the vertices of cycles $C_1,\ldots,C_{3m}$ to the vertices of cycles $\tilde{C_1},\ldots,\tilde{C_m}$ in exactly the same way as $\varphi$ mapped these vertices in the \NP-hardness proof of \LBHom{}. 
It is a routine exercise to show that $G'$ has pathwidth at most~$4$ 
and that $H'$ has pathwidth at most~$3$.

\begin{figure}[htb]
\centering
\includegraphics[scale=.85]{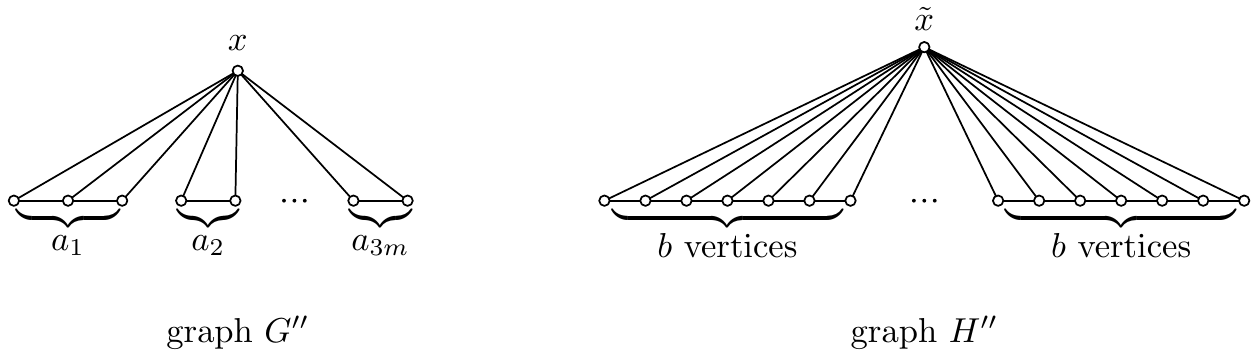}
\caption{A schematic illustration of the graphs $G''$ and $H''$ that are constructed from a given instance $(A,b)$ of $3$-{\sc Partition} in the proof of statement (iii) in Theorem~\ref{t-np}.
}
\label{fig:reduction_surjective}
\end{figure}

The reduction for \LIHom{} is even easier; see Figure~\ref{fig:reduction_surjective}. Given an instance $(A,b)$ of $3$-{\sc Partition}, we create a graph $G''$ by adding a universal vertex $x$ to the disjoint union of $3m$ paths 
on $a_1,a_2,\ldots,a_{3m}$ vertices, respectively. Graph $H''$ is obtained from the disjoint union of $m$ paths on $b$ vertices by adding a universal vertex $\tilde{x}$. It is easy to verify that there exists a locally injective 
homomorphism $\varphi''$ from $G''$ to $H''$, mapping $x$ to $\tilde{x}$ and all other vertices of $G''$ to the vertices of degree~$2$ or $3$ in $H''$, if and only if $(A,b)$ is a yes-instance of $3$-{\sc Partition}. 
The observation that both $G''$ and $H''$ have pathwidth~$2$ completes the proof of Theorem~\ref{t-np}.
\qed
\end{proof}

\medskip
We now consider the case where we bound the maximum degree of $G$ instead of the treewidth of $G$. 
We will combine some known results in order to show that bounding the maximum degree of $G$ does not yield tractability for any of our three problems \LBHom, \LIHom{} and \LSHom. We first introduce some additional terminology.
An {\it equitable partition} of a connected graph $G$ is a partition of its vertex set in blocks  $B_1,\ldots, B_k$ such that any vertex in $B_i$
has the same number $m_{i,j}$ of neighbors in $B_j$. We call the matrix $M=(m_{i,j})$ corresponding to the coarsest equitable partition of $G$
(in which the blocks are ordered in some canonical way; cf.~\cite{An80}) the {\it degree refinement matrix} of $G$, denoted as $\drm(G)$. We will use the following lemma; a proof of the first statement in this lemma can be found in the paper of Fiala and Kratochv\'{\i}l~\cite{FK02}, whereas the second statement is due to Kristiansen and Telle~\cite{KT00}.

\begin{lemma}\label{l-drm}
Let $G$ and $H$ be two graphs. Then the following two statements hold:
\begin{itemize}
\item [(i)] if $G\parc H$ and $\drm(G)=\drm(H)$, then $G\fullc H$;
\item [(ii)] if $G\surjc H$ and $\drm(G)=\drm(H)$, then $G\fullc H$.
\end{itemize}
\end{lemma}

Kratochv\'{\i}l and K\v{r}iv\'anek~\cite{KK88} showed that \xLBHom{K_4} is \NP-complete, where $K_4$ denotes the complete graph on four vertices. 
Since a graph $G$ allows a locally bijective homomorphism to $K_4$ only if $G$ is $3$-regular, \xLBHom{K_4} is \NP-complete on $3$-regular graphs. The degree refinement matrix of a $3$-regular graph is the $1\times 1$ matrix whose only entry is~$3$. Consequently, due to Lemma~\ref{l-drm}, \xLBHom{K_4} is equivalent to \xLIHom{K_4} and to \xLSHom{K_4} on $3$-regular graphs. This yields the following result.

\begin{theorem}
\label{t-degree}
The problems \LBHom{}, \LIHom{} and \LSHom{} are \NP-complete on input pairs $(G,K_4)$ where $G$ has maximum degree~$3$.
\end{theorem}

Theorem~\ref{t-degree} is tight in the following sense. All three problems \LBHom{}, \LIHom{} and \LSHom{} are polynomial-time solvable on input pairs $(G,H)$ where $G$ has maximum degree at most~$2$. Moreover, the first two problems are also polynomial-time solvable on input graphs $(G,H)$ where only $H$ has maximum degree at most~$2$. This does not hold for the \LSHom{}, as \xLSHom{K_3} is \NP-complete~\cite{KT00}.

\section{Polynomial-Time Results}
\label{s-pol}

In Section~\ref{s-npcom}, we showed that \LBHom{}, \LIHom{} and \LSHom{} are \NP-complete when either the treewidth or the maximum degree of the guest graph is bounded. In this section, we show that all three problems become polynomial-time solvable if we bound both the treewidth and the maximum degree of $G$. For the problems \LBHom{} and \LIHom{}, our polyno\-mial-time result follows from
reformulating these problems as constraint satisfaction problems and applying a result of Dalmau et al.~\cite{DKV02}.
In order to explain this, we need some additional terminology.

A relational structure $(A,R_1,\dots,R_k)$ is a finite set $A$, 
called the {\em base set}, together with a collection of relations $R_1,\dots,R_k$. The arities
of these relations determine the vocabulary of the structure.
A homomorphism between two relational structures of the same vocabulary is a mapping between the base sets such that 
all the relations are preserved.

Fiala and Kratochv\'{\i}l~\cite{FK06} observed that locally injective and locally bijective homomorphisms between graphs can be expressed as homomorphisms between relational structures as follows. A locally injective homomorphism $f:G\to H$ can be expressed as a homomorphism between relational structures 
$(V_G,E_G,E'_G)$ and $(V_H,E_H,E'_H)$, where the new binary relation $E'$ consists of pairs of distinct vertices that have at least one common neighbor.
Since $f$ maps distinct neighbors of a vertex $v$ to distinct neighbors of $f(v)$, we get that $f$ is a homomorphism of the associated relational structures. On the other hand, 
if $(V_G,E_G,E'_G)$ and $(V_H,E_H,E'_H)$ are constructed from $G$ and $H$ as described above, and if $f$ is a homomorphism between them, then the relations $E'_G$ 
and $E_H'$ guarantee that no two vertices with a common neighbor 
in $G$ are mapped to the same target in $H$.
In other words, $f$ is a locally injective homomorphism between the graphs $G$ and $H$. 
An analogous construction works for locally bijective homomorphisms. Here, we need to express $G$ using two binary relations $E_G$ and $E'_G$ as above, together with $\Delta(G)+1$ unary relations. A unary relation can be viewed as a set: here, the $i$-th set will consists of all vertices of degree~$i-1$. These unary relations guarantee that degrees are preserved, and consequently that the associated graph homomorphisms are locally bijective.

The \emph{Gaifman graph} $\cG_{\cal A}$ of a relational structure ${\cal A}=(A,R_1,\dots,R_k)$ is 
the graph with vertex set $A$, where 
any two distinct vertices $u$ and $v$
are joined by an edge if they are bound by some relation. Formally
$u,v\in E_{\cG_{\cal A}}$ if and only if for some relation $R_i$ of arity $r$ and $(a_1,\dots,a_r )\in R_i$ 
it holds that $\{u,v\}\subseteq \{a_1,\dots.a_r\}$.

As a direct consequence of a result of Dalmau et al.~\cite{DKV02}, the existence of a homomorphism 
between two relational structures ${\cal A}$ and ${\cal B}$  can be decided in polynomial time if the treewidth of $\cG_{\cal A}$ is bounded by a constant. 
This leads to Theorem~\ref{thm:bounded} below. 

\begin{theorem}\label{thm:bounded}
The problems \LBHom{} and \LIHom{} can be solved in polynomial time when $G$ has bounded treewidth and $G$ or $H$ has bounded maximum degree.
\end{theorem}

\begin{proof}
First suppose that $G$ has bounded treewidth and bounded maximum degree. 
Observe that for locally injective and locally bijective homomorphisms, the Gaifman graph $\cG_{\cal A}$ is isomorphic to $G^2$, which is the graph arising from $G$ by adding an edge between any two vertices at distance~$2$. It suffices to observe that $\tw(G^2)\leq \Delta(G)(\tw(G)+1)-1$, as we can transform any tree decomposition $T$ of $G$ of width $\tw(G)$ into a desired tree decomposition of $G^2$ by adding to each node $X$ of $T$ all the neighbors of every vertex from $X$.
Since $G \parc H$ implies that $\Delta(H)\ge \Delta(G)$, the theorem also holds if we bound the maximum degree of $H$ instead of~$G$.
\qed
\end{proof}

To our knowledge, locally surjective homomorphisms have not yet been expressed as homomorphisms between relational structures.
Hence, in the proof of Theorem~\ref{t-bounded} below, we present a polynomial-time algorithm for \LSHom{} when $G$ has bounded treewidth and bounded maximum degree. We first introduce some additional terminology. 

Let $\varphi$ be a locally surjective homomorphism from $G$ to $H$. 
Let $v\in V_G$ and $p\in V_H$. If $\varphi(v)=p$, i.e., if $\varphi$ maps vertex $v$ to color $p$, then we say that $p$ is {\em assigned} to $v$. By definition, for every vertex $v\in V_G$, the set of colors that are assigned to the neighbors of $v$ in $G$ is exactly the neighborhood of $\varphi(v)$ in $H$. Now suppose we are given a homomorphism $\varphi'$ from an induced subgraph $G'$ of $G$ to $H$. For any vertex $v\in V_{G'}$, we say that $v$ {\em misses} a color $p\in V_H$ if $p\in N_H(\varphi'(v))\setminus \varphi(N_{G'}(v))$, i.e., if $\varphi'$ does not assign $p$ to any neighbor of $v$ in $G'$, but any locally surjective homomorphism $\varphi$ from $G$ to $H$ that extends $\varphi'$ assigns $p$ to some neighbor of $v$ in~$G'$.

Let $T$ be a nice tree decomposition of $G$ rooted in $R$. For every node $X\in V_T$, we define $G_X$ to be the subgraph of $G$ induced by the vertices of $X$ together with the vertices of all the nodes that are descendants of $X$. 
In particular, we have $G_R=G$.

\begin{definition}
\label{def:feasible}
Let  $X\in V_T$, and let $c:X\rightarrow V_H$ and $\mu:X\rightarrow 2^{V_H}$ be two mappings. The pair $(c,\mu)$ is {\em feasible} for $G_X$ if there exists a homomorphism $\varphi$ from $G_X$ to $H$ satisfying the following three conditions:
\begin{itemize}
\item[(i)] $c(v)=\varphi(v)$ for every $v\in X$;
\item[(ii)] $\mu(v)=N_H(\varphi(v))\setminus \varphi(N_{G_X}(v))$ for every $v\in X$;
\item[(iii)] $\varphi(N_G(v))=N_H(\varphi(v))$ for every $v\in V_{G_X}\setminus X$.
\end{itemize}
\end{definition}

In other words, a pair $(c,\mu)$ consists of a coloring $c$ of the vertices of $X$, together with a collection of sets $\mu(v)$, one for each $v\in X$, consisting of exactly those colors that $v$ misses. 
Informally speaking, a pair $(c,\mu)$ is feasible for $G_X$ if there is a homomorphism $\varphi: G_X\to H$ 
such that $\varphi$ ``agrees'' with the coloring $c$ on the set $X$, and such that none of the vertices in $V_{G_X}\setminus X$ misses any color. The idea is that if a pair $(c,\mu)$ is feasible, then 
such a homomorphism $\varphi$ might have an extension $\varphi^*$ that is a locally surjective homomorphism from $G$ to $H$. After all, for any vertex $v\in X$ that misses a color when considering $\varphi$, this color might be assigned by $\varphi^*$ to a neighbor of $v$ in the set $V_G\setminus V_{G_X}$.

We now prove a result for \LSHom{} similar to Theorem~\ref{thm:bounded}. 

\begin{theorem}\label{t-bounded}
The problem \LSHom{} can be solved in polynomial time when $G$ has bounded treewidth and $G$ or $H$ has bounded maximum degree.
\end{theorem}

\begin{proof}
Let $(G,H)$ be an instance of \LSHom{} such that the treewidth of the guest graph $G$ is bounded. Throughout the proof, we assume that the maximum degree of $H$ is bounded, and show that the problem can be solved in polynomial time under these restrictions. Since $G \surjc H$ implies that $\Delta(G)\ge \Delta(H)$, our polynomial-time result applies also if we bound the maximum degree of $G$ instead of~$H$.

We may assume without loss of generality that both $G$ and $H$ are connected, as otherwise we just consider all pairs $(G_i,H_j)$ separately, where $G_i$ is a connected component of $G$ and $H_j$ is a connected component of $H$.
Because $G$ has bounded treewidth, we can compute a tree decomposition of $G$ of width $\tw(G)$ in linear time using Bodlaender's algorithm~\cite{Bo96}. We transform this tree decomposition into a nice tree decomposition $T$ of $G$ with width $\tw(G)$
with at most $4|V_G|$ nodes using the linear-time algorithm of Kloks~\cite{Kloks94}. 
Let $R$ be the root of $T$ and let $k=\tw(G)+1$.

For each node $X\in V_T$, let $F_X$ be the set of all feasible pairs $(c,\mu)$ for $G_X$. For every feasible pair $(c,\mu)\in F_X$ and every $v\in X$, it holds that $\mu(v)$ is a subset of $N_H(c(v))$. Since $|X|\leq k$ and $|N_H(c(v))|\leq \Delta(H) k$ for every $v\in X$ and every mapping $c:X\rightarrow V_H$, this implies that $|F_X|\le |V_H|^{k} 2^{\Delta(H) k}$ for each $X\in V_T$. As we assumed that both $k$ and $\Delta(H)$ are bounded by a constant, the set $F_X$ is of polynomial size with respect to $|V_H|$. 

The algorithm considers the nodes of $T$ in a bottom-up manner, starting with the leaves of $T$ and processing a node $X\in V_T$ only after its children have been processed. For every node $X$, the algorithm computes the set $F_X$ in the way described below. We distinguish between four different cases. The correctness of each of the cases easily follows from the definition of a locally surjective homomorphism and Definition~\ref{def:feasible}.

\begin{enumerate}
\item {\em $X$ is a leaf node of $T$.} 
We consider all mappings $c:X\rightarrow V_H$. For each mapping $c$, we check whether $c$ is a homomorphism from $G_X$ to $H$. If not, then we discard $c$, as it can not belong to a feasible pair due to condition (i) in Definition~\ref{def:feasible}. For each mapping $c$ that is not discarded, we compute the unique mapping $\mu$ satisfying $\mu(v)=N_H(c(v))\setminus c(N_{G_X}(v))$ for each $v\in X$, and we add the pair $(c,\mu)$ to $F_X$. It follows from condition (ii) that the obtained set $F_X$ indeed contains all feasible pairs for $G_X$. As there is no vertex in $V_{G_X}\setminus X$, every pair $(c,\mu)$ trivially satisfies condition (iii). The computation of $F_X$ can be done in $O(|V_H|^{k} k (\Delta(H)+k))$ time in this case.

\item {\em $X$ is a forget node.} Let $Y$ be the child of $X$ in $T$, and let $\{u\}=Y\setminus X$. Observe that $(c,\mu)\in F_X$ if 
and only if there exists a feasible pair $(c',\mu')\in F_Y$ such that $c(v)=c'(v)$ and $\mu(v)=\mu'(v)$ for every $v\in X$, and $\mu'(u)=\emptyset$.
Hence we examine each $(c',\mu')\in F_Y$ and check whether $\mu'(u)=\emptyset$ is satisfied. If so, we first restrict $(c',\mu')$ on $X$ to get $(c,\mu)$ and then we insert the obtained feasible pair into $F_X$.
This procedure needs $O(|F_Y| k\Delta(H))$ time in total.

\item {\em $X$ is an introduce node.} Let $Y$ be the child of $X$ in $T$, and let $\{u\} =X\setminus Y$. Observe that $(c,\mu)\in F_X$ if 
and only if there exists a feasible pair $(c',\mu')\in F_Y$ such that, for every $v\in Y$, it holds that $c(v)=c'(v)$, $\mu(v)=\mu'(v)\setminus c(u)$ if $uv\in E_G$, and $\mu(v)=\mu'(v)$ if $uv\notin E_G$. Hence, for each $(c',\mu')\in F_Y$, we consider all $|V_H|$ mappings $c:X\rightarrow V_H$ that extend $c'$. 
For each such extension $c$, we test whether $c$ is a homomorphism from $G_X$ to $H$ by checking the adjacencies of $c(u)$ in $H$. 
If not, then we may safely discard $c$ due to condition (i) in Definition~\ref{def:feasible}. Otherwise, we compute the unique mapping $\mu:X\rightarrow 2^{V_H}$ satisfying
$$
\mu(v)=\begin{cases}
N_H(c(u))\setminus c(N_{G_X}(u)) & \text{if } v=u \\
\mu'(v)\setminus c(u) & \text{if $v\neq u$ and $uv\in E_G$}\\
\mu'(v) & \text{if $v\neq u$ and $uv\notin E_G$} \; ,
\end{cases}
$$
and we add the pair $(c,\mu)$ to $F_X$; due to condition (ii), this pair $(c,\mu)$ is the unique feasible pair containing $c$. Computing the set $F_X$ takes at most 
$O(|F_Y||V_H| k\Delta(H))$ time in total.

\item {\em $X$ is a join node.} Let $Y$ and $Z$ be the two children of $X$ in $T$. 
Observe that $(c,\mu)\in F_X$ if 
and only if there exist 
feasible pairs $(c_1,\mu_1)\in F_Y$ and $(c_2,\mu_2)\in F_Z$ such that, for every $v\in X$,  
$c(v)=c_1(v)=c_2(v)$ and $\mu(v) =\mu_1(v)\cap \mu_2(v)$.
Hence the algorithm considers every combination of $(c_1,\mu_1)\in F_Y$ with $(c_2,\mu_2)\in F_Z$ 
and if they agree on the first component $c$, the other component $\mu$ is determined uniquely by taking the intersection of $\mu_1(v)$ and $\mu_2(v)$ for every $v\in X$.
This procedure computes the set $F_X$ in $O(|F_Y||F_Z| k\Delta(H))$ time in total.
\end{enumerate}

Finally, observe that a locally surjective homomorphism from $G$ to $H$ exists if and only if there exists a feasible pair $(c,\mu)$ for $G_R$ such that $\mu(v)=\emptyset$ for all $v\in R$.
Since $T$ has at most $4|V_G|$ nodes, we obtain a total running time of 
$O(|V_G| (|V_H|^{k} 2^{\Delta(H)k})^2 k\Delta(H))$.
As we assumed that both $k=\tw(G)+1$ and $\Delta(H)$ are bounded by a constant, our algorithm runs in polynomial time.
\qed
\end{proof}

Note that Theorem~\ref{thm:bounded} can be derived by solving \LIHom{} using a dynamic programming approach that strongly resembles the one for \LSHom{} described in the proof of Theorem~\ref{t-bounded}, 
together with the fact that $(G,H)$ is a yes-instance of \LBHom{} if and only if it is a yes-instance for both \LIHom{} and \LSHom{}. In a dynamic programming algorithm for solving \LIHom{}, instead of keeping track of sets $\mu(v)$ of colors that a vertex $v\in X$ is missing, we keep track of sets $\alpha(v)$ of colors that have already been assigned to the neighbors of a vertex $v\in X$. This is because in a locally injective homomorphism from $G$ to $H$, no color may be assigned to more than one neighbor of any vertex. 
In this way we can adjust Definition~\ref{def:feasible} in such a way that it works for locally injective instead of locally surjective homomorphisms. We omit further details, but we expect that 
a dynamic programming algorithm of this kind 
will have smaller hidden constants in the running time estimate  than the more general method of Dalmau et al.~\cite{DKV02}.

We conclude this section with one more polynomial-time result.
It is known that the problems \LBHom{} and \LSHom{} are polynomial-time solvable when $G$ is a tree~\cite{FP10}, and consequently when $G$ has treewidth~$1$. We claim that the same holds for the \LIHom{} problem.

\begin{theorem}
\label{t-polytree}
The \LIHom{} problem can be solved in polynomial time when $G$ has treewidth~$1$.
\end{theorem}

\begin{proof}
Let us first state some terminology and useful known results.
The universal cover $T_G$ of a connected
graph $G$ is the unique tree (which may have an infinite number of vertices) 
such that there is a locally bijective homomorphism from $T_G$ to $G$.
One way to define this mapping is as follows.
Consider all finite walks in $G$ that start from an arbitrary fixed vertex in $G$ and that do not
traverse the same edge in two consecutive steps. Each such walk will correspond to a vertex of $T_G$.
We let two vertices of $T_G$ be adjacent
if and only if one can be obtained from the other by deleting the last vertex of the walk.
Then the mapping $f_G$  that maps every walk to its last vertex is a locally bijective homomorphism from $T_G$ to $G$~\cite{An80}.
It is also known that $T_G=G$ if and only if $G$ is a tree~\cite{An80}. Moreover, for any two graphs $G$ and $H$, $G\parc H$ implies that $T_G\parc T_H$~\cite{FM06}.

Now let $(G,H)$ be an instance of \LIHom{} where $G$ has treewidth~$1$. We assume, without loss of generality, that both $G$ and $H$ are connected. In particular, $G$ is a tree. We claim that $G\parc H$ if and only if $T_G\parc T_H$. The forward implication follows from above. To show the backward implication, suppose that $T_G\parc T_H$. Then $G\parc T_H$, because $T_G=G$. 
Let $f$ be a locally injective homomorphism from $G$ to $T_H$. Then, because $G\parc T_H$ and $T_H\fullc H$, we have $G\parc H$. To explain this, consider
the mapping $f':V_G\to V_H$ defined by $f'(u)=f_H(x)$ if and only if $f(u)=x$. Notice that $f'$ is a locally injective homomorphism from $G$ to $H$.
The desired result follows from this claim combined with the fact that we can check in polynomial time whether $T_G\parc T_H$ holds for two graphs $G$ and $H$~\cite{FP10}.
\qed
\end{proof}

\section{Conclusion}\label{s-con}

Theorem~\ref{t-polytree} states that \LIHom{} can be solved in polynomial time when the guest graph has treewidth~1, while Theorem~\ref{t-np} implies that the problem is \NP-complete when the guest graph has treewidth~2. This shows that the bound on the pathwidth in the third statement of Theorem~\ref{t-np} is best possible. We leave it as an open problem to determine whether the bounds on the pathwidth in the other two statements of Theorem~\ref{t-np} can be reduced further.

We conclude this paper with some remarks on the parameterized complexity of the problems \LIHom{},  \LSHom{} and \LBHom{}. The hardness results in this paper show that all three problems are para-\NP-complete when parameterized by either the treewidth of $G$ or the maximum degree of $G$. Theorems~\ref{thm:bounded} and~\ref{t-bounded} show that the problems are in {\sf XP} when parameterized jointly by the treewidth of $G$ and the maximum degree of $G$. A natural question is whether the problems are {\sf FPT} when parameterized by the treewidth of $G$ and the maximum degree of $G$, i.e., whether they can be solved in time $f(\tw(G),\Delta(G))\cdot (|V_G|+|V_H|)^{O(1)}$ for some function~$f$ that does not depend on the sizes of $G$ and $H$.

\medskip
\noindent
{\small
{\bf Acknowledgements.} We would like to thank Isolde Adler for posing the research
questions that we addressed in our paper and for helpful discussions. The fourth
author also thanks Jan Arne Telle for fruitful discussions.}

\end{document}